\documentclass[letterpaper]{article} %

\usepackage{booktabs}
\usepackage{amsmath}
\usepackage{svg}
\usepackage{aaai2027}
\usepackage{times}  %
\usepackage{helvet}  %
\usepackage{courier}  %
\usepackage[hyphens]{url}  %
\usepackage{graphicx} %
\usepackage{natbib}  %
\usepackage{caption} %
\usepackage{algorithm}
\usepackage{algorithmic}

\usepackage{newfloat}
\usepackage{listings}
\usepackage{xspace}
\newcommand{\LTLf}{\ensuremath{\textsf{LTL}_f}\xspace}

\newcommand{\Next}{\ensuremath{\textsf{X}}\xspace}
\newcommand{\Until}{\ensuremath{\textsf{U}}\xspace}

\usepackage{amsmath}

\newcommand{\A}{\ensuremath{\mathcal{A}}\xspace}

\newcommand{\LTL}{\ensuremath{\text{LTL}}\xspace}

\newcommand{\LTLNext}{\textsf{X}\ }
\newcommand{\LTLWNext}{\textsf{X}_w\ }

\newcommand{\LTLUntil}{\ \textsf{U}\ }

\newcommand{\LTLAlways}{\textsf{G}\ }
\newcommand{\LTLEventually}{ \textsf{F}\  }

\newcommand{\naf}{\mathop{not}}

\DeclareCaptionStyle{ruled}{labelfont=normalfont,labelsep=colon,strut=off} %
\floatstyle{ruled}
\newfloat{listing}{tb}{lst}{}
\floatname{listing}{Listing}
\usepackage{amsthm}
\usepackage{amssymb}
\usepackage{thmtools}

\declaretheoremstyle[
  spaceabove=6pt, spacebelow=6pt,
  qed=$\blacksquare$
]{defstyle}

\declaretheoremstyle[
  spaceabove=6pt, spacebelow=6pt,
  qed=$\triangleleft$
]{exstyle}

\declaretheorem[style=defstyle, name=Definition]{definition}
\declaretheorem[style=exstyle,  name=Example]{example}
\declaretheorem[name=Proposition]{proposition}
\declaretheorem[name=Lemma]{lemma}
\declaretheorem[name=Corollary]{corollary}

\renewcommand{\A}{\mathcal{A}}
\newcommand{\Vars}{\mathsf{Vars}}
\newcommand{\Lits}{\mathsf{Lits}}
\renewcommand{\L}[1]{\mathsf{L}(#1)}
\newcommand{\base}[1]{\mathsf{var}(#1)}
\newcommand{\opplit}[1]{\overline{#1}}
\newcommand{\Comp}[2]{\mathsf{Comp}_{#1}(#2)}
\newcommand{\Mod}[1]{\mathsf{Mod}(#1)}
\newcommand{\asgn}[1]{\alpha_{#1}}
\newcommand{\traceof}[1]{\pi_{#1}}
\newcommand{\agrees}[3]{#1 \equiv_{#2} #3}
\newcommand{\featof}[1]{\mathsf{feat}(#1)}
\newcommand{\litof}[2]{\mathsf{lit}_{#1}(#2)}

\usepackage{comment}
\newif\ifappendix
\appendixtrue
\nocopyright

\ifappendix
  \newenvironment{appendixonly}{}{}
\else
  \excludecomment{appendixonly}
\fi

\title{Formal Explanations for Linear Temporal Logic over Finite Traces}

\author{
    Valeria Fionda,
    Antonio Ielo,
    Francesco Ricca
}
\affiliations{
    University of Calabria, Italy\\
    valeria.fionda@unical.it, antonio.ielo@unical.it, francesco.ricca@unical.it
}

\usepackage{bibentry}

\begin{document}

\maketitle

\begin{abstract}
Linear Temporal Logic over Finite Traces (\LTLf) is widely used in Artificial Intelligence to specify desired temporal behavior and determine whether finite traces satisfy or violate it. Explaining these outcomes is essential for debugging and trust. However, existing approaches typically capture either sufficient reasons or counterfactual causes, without clarifying their relationship. 
We introduce a native characterization of abductive (AXps) and contrastive (CXps) explanations directly over finite traces, based on \emph{locking} and \emph{unlocking} time-indexed events. We prove its equivalence to the classifier-based formulation, and show that CXps coincide with a recently proposed temporal causality semantics for \LTLf. 
This characterization enables an efficient implementation based on the enumeration of minimal unsatisfiable cores and minimal correction sets of ASP programs. Experiments on real-world event logs demonstrate the effectiveness of the approach.
\end{abstract}

\begin{figure*}
    \centering
    \includegraphics[width=.75
    \linewidth]{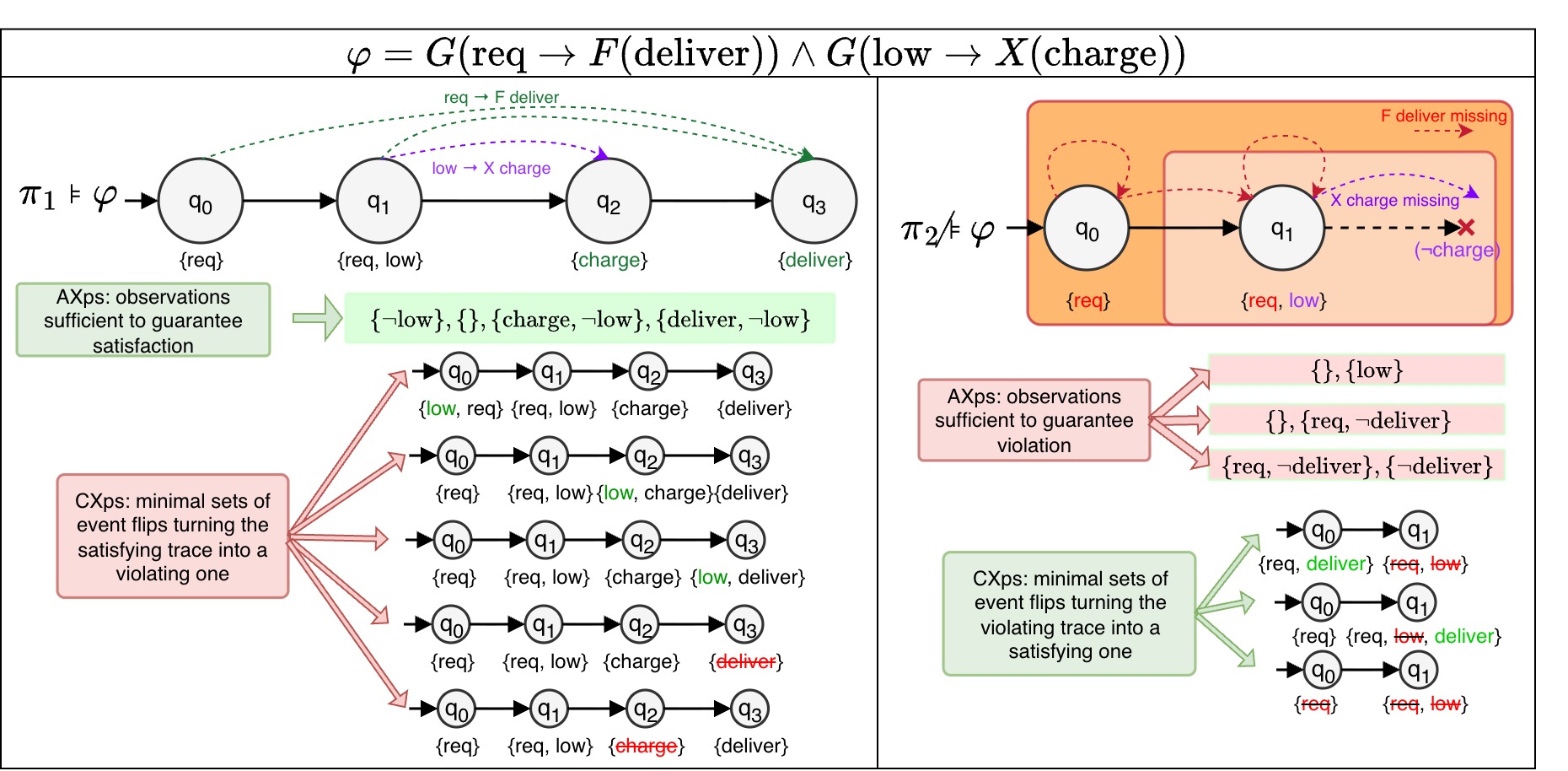}
    \caption{AXps (sufficient reasons) and CXps (minimal outcome-flipping interventions) for $\varphi$ on a compliant trace $\pi_1$ and a violating trace $\pi_2$. Sets are listed state by state; e.g. $\{\},\{low\}$ denotes $\{low_1\}$.}
    \label{fig:motivating}
\end{figure*}

\section{Introduction}
Linear Temporal Logic over Finite Traces~\cite{DBLP:conf/ijcai/GiacomoV13} (\LTLf) is a popular formalism for temporal reasoning in AI. It has been successfully applied to several domains, such as reinforcement learning~\cite{de2019foundations}, planning~\cite{bacchus1998planning, GiacomoR18}, robotics~\cite{kress2009temporal} and business process management~\cite{DBLP:journals/ife/AalstPS09,de2014monitoring}. 

In applications, \LTLf formulas are used to describe \emph{desired behavior}. A formula may describe the dynamics of a system, or the goals driving an agent's behavior. Traces (interpretations) represents states of a system over time, or a sequence of actions performed by an agent. 
When a trace $\pi$ satisfies a formula $\varphi$, we say that $\pi$ \emph{complies} with the specification; otherwise, it \emph{violates} it. In either case, being able to explain \emph{why} is important.

Figure~\ref{fig:motivating} illustrates the explanation problem via an autonomous delivery agent subject to
$
\varphi =
\LTLAlways(\mathit{req} \rightarrow \LTLEventually\mathit{deliver})
\wedge
\LTLAlways(\mathit{low} \rightarrow \LTLNext\mathit{charge}),
$
which requires every request to be eventually served and every low-battery to be immediately followed by charging. The trace $\pi_1$ (of Figure~\ref{fig:motivating}) satisfies $\varphi$, but this alone does not explain \emph{why}. We focus on two complementary kinds of explanation. A \emph{sufficient reason} is a set of time-indexed events and non-events that guarantees satisfaction: for $\pi_1$, the absence of $\mathit{low}$ throughout the trace already guarantees satisfaction of the charging obligation, no matter how the remaining observations are altered. A \emph{minimal intervention} is a minimal set of observations whose change flips the outcome: adding $\mathit{low}$ to the first time instant of $\pi$ not followed by $\mathit{charge}$ turns it into a violating trace. Conversely, $\pi_2$ violates $\varphi$: here a sufficient reason guarantees the violation  (e.g., knowing only that $\mathit{low}$ holds at the second time instant and $\mathit{charge}$ is absent at the third time instant makes \emph{every} trace agreeing with these two observations violate $\varphi$), while minimal interventions restore satisfaction. Thus the explanation problem answers two complementary questions: \emph{why does the trace satisfy or violate the specification?} and \emph{what must change to obtain the opposite outcome?}

These two kinds of explanation are precisely the central notions of formal explainable AI (FXAI)~\cite{DBLP:conf/aaai/IgnatievNM19,DBLP:conf/ijcai/ShihCD18}: \emph{abductive} (AXps) and \emph{contrastive} (CXps) explanations. Given a classifier and an input instance, an abductive explanation is a subset of feature values sufficient to guarantee the predicted class, whereas a contrastive explanation is a subset of features whose change yields a different class. Importantly, subset-minimal abductive and contrastive explanations are related by hitting-set duality~\cite{joao,reiter1987theory}, which also supports effective enumeration algorithms~\cite{DBLP:journals/constraints/LiffitonPMM16}.
In principle, an \LTLf formula can be cast as a binary classifier over time-indexed propositions, making these notions formally applicable. However, this encoding is not the usual way temporal formulas are understood and makes it harder to see explanations as observations that are kept fixed or allowed to change. Thus, a \emph{native characterization} of AXps and CXps for \LTLf is missing. Existing approaches capture either sufficient reasons for satisfaction~\cite{ferrere2015trace} or counterfactual causes of violation~\cite{meng2025causal}, but none connects the two perspectives or grounds them in hitting-set duality.
We fill this gap through the simple operations of \emph{locking} and \emph{unlocking} temporal observations: locked observations play the role of the feature values an AXp keeps fixed, whereas unlocked ones are those a CXp is allowed to alter.
AXps and CXps are thus defined directly over finite traces, with no classifier-based reformulation.
Our contributions are threefold: \textit{(i)} we characterize AXps and CXps for the satisfaction and violation of \LTLf specifications, directly over finite traces; \textit{(ii)} we prove that this characterization is equivalent to the classifier-based one, establish the hitting-set duality between AXps and CXps for \LTLf, and show that CXps coincide with the temporal causes of~\cite{meng2025causal} under fixed-length interventions;
\textit{(iii)} we reduce their computation to the enumeration of minimal unsatisfiable subsets and minimal correction sets, implement the approach in Answer Set Programming (ASP)~\cite{DBLP:journals/ngc/GelfondL91,DBLP:journals/cacm/BrewkaET11}, and evaluate it on real-world event logs.

\section{Related Work}

\paragraph{Formal Explanations.}
Our work instantiates the general framework of abductive and contrastive explanations~\cite{DBLP:conf/ijcai/ShihCD18,DBLP:conf/aaai/IgnatievNM19,joao,DBLP:journals/jolli/DarwicheH23} in the setting of \LTLf. 
However, the standard formulation treats the object being explained as a classifier over features. While this viewpoint is formally applicable to \LTLf, it is not natural for temporal-logic practitioners. We therefore provide an equivalent characterization stated directly in terms of \LTLf formulas and finite traces.
Similar specializations have been developed for settings in which the object being explained is not naturally represented as a classifier, such as ranking functions~\cite{DBLP:conf/jelia/ChiarielloM25} and finite-state machines~\cite{DBLP:journals/corr/abs-2602-13351}. Our work follows the same direction by providing a specialization for \LTLf formulas.

\paragraph{Temporal Implicants \& Temporal Causality.}
Implicants are partial interpretations guaranteeing a propositional formula’s truth or falsity and, for propositional binary classifiers, coincide with abductive explanations~\cite{joao}. 
In LTL, temporal implicants have only been studied under infinite-trace semantics~\cite{ferrere2015trace}, and implemented with a greedy algorithm computing a single, possibly non-minimal implicant. 
The connection to unsatisfiable cores is mentioned as a related problem but not further developed. We instead address finite traces, recast the problem in the FXAI framework, and provide an efficient implementation based on unsatisfiable cores and correction sets.

A recent work by~\cite{meng2025causal} proposes a revised version of \emph{cause} for \LTLf, specializing the Halpern's~\cite{DBLP:books/daglib/0044508}. It extends earlier works~\cite{DBLP:conf/cav/BeerBCOT09} that considers only singleton causes, to conjunctive causes. Neither line addresses the dual notion of sufficient reasons for satisfaction or violation. We will prove that this notion is closely related to CXps.

\paragraph{Unsatisfiable Cores in Temporal Reasoning.} 
Minimal unsatisfiable cores and correction sets~\cite{LiSa08,DBLP:journals/ai/Marques-SilvaJM17} are at the base of our implementation. 
\LTLf admits notions of unsatisfiable cores~\cite{DBLP:journals/jair/RoveriCFG24}, whose computation has been the focus of recent works~\cite{DBLP:journals/jair/RoveriCFG24,DBLP:journals/logcom/NiuXZLHS24,DBLP:conf/aaai/FiondaIR26,DBLP:conf/aaai/IeloMPR26}.  
Unsatisfiable cores explain \emph{why a formula is unsatisfiable}, whereas abductive explanations answer \emph{why a trace violates/satisfies a formula}. The first is a form of \emph{global} explanation, the latter is a form of \emph{local} explanation~\cite{DBLP:journals/csur/GuidottiMRTGP19,DBLP:conf/aaai/0001I22}.
Our implementation is based on bounded \LTLf satisfiability~\cite{lpnmmr22}, obtaining AXps and CXps as unsatisfiable cores and correction sets of a logic program~\cite{DBLP:journals/ai/AlvianoDFPR23}. The underlying solver is based on SAT techniques~\cite{DBLP:journals/constraints/LiffitonPMM16}.

\section{Answer Set Programming}
An \emph{atom} is an expression of the form $p(t_1,\ldots,t_n)$ where $p$ is a predicate, $t_1,\ldots,t_n$ are terms (i.e., constants, integer numbers, strings or variables), and $n$ is known as \emph{arity} of the predicate. 
A \emph{literal} is either an atom $a$ or its \emph{negation} $\naf a$, where $\naf$ denotes the negation as failure. 
A \emph{rule} $r$ is an expression of the form $h \leftarrow b_1,\ldots,b_n$ where $h$ is an atom referred to as \emph{head}, denoted by $H_r$, that can also be omitted, $n\geq 0$, and $b_1,\ldots,b_n$ is a conjunction of literals referred to as \emph{body}, denoted by $B_r$. 
In particular, a rule is said to be a \emph{constraint} if its head is omitted, while it is said to be a \emph{fact} if $n=0$.
A rule $r$ is \emph{safe} if each variable appears at least in one positive literal in the body of $r$. 
A \emph{program} in ASP is a finite set of safe rules.
Given a program $P$, and $r \in P$, $ground(r)$ is the set of ground instantiations of $r$ obtained by replacing variables in $r$ with constants in $P$; whereas $ground(P)$ is the union of ground instantiations of rules in $P$. 
Given a program $P$, the Herbrand's base of $P$, denoted by $\mathcal{B}_P$, is the set of atoms constructible from constants and predicate names occurring in $P$. An \textit{interpretation} $I$ for $P$ is a subset of $\mathcal{B}_P$, which is an answer set of $P$ if (i) $I$ is a model, i.e., for each rule $r \in ground(P)$ either the head of $r$ is true wrt $I$ or the body of $r$ is false wrt $I$; and (ii) $I$ is a minimal model of its GL-reduct~\cite{DBLP:journals/ngc/GelfondL91}.%
\paragraph{Minimal Unsatisfiable Subprograms.} %
Consider a program $P$ and a set of objective atoms $O \subseteq \mathcal{B}_P$. %
For $S \subseteq O$, $\text{enforce}(P,O,S)$ is the program obtained from $P$ by adding a choice over atoms in $O$ (i.e., 
$o_i \texttt{:- } no_i$
$no_i \texttt{:- } o_i$, for $i\in [1,n]$)
and a set of constraints of the form $\texttt{:- } \naf o$, for every $o \in S$.
Intuitively, $\text{enforce}(P,O,S)$ augments the program $P$ in such a way that the objective atoms can be arbitrarily chosen (i.e. either as true or false) but the atoms in $S$ are \emph{enforced} to be true. An \emph{unsatisfiable subset} for $P$ wrt\ the set of objective atoms $O$ is a set of atoms $U \subseteq O$ such that $\text{enforce}(P,O,U)$ is incoherent~\cite{DBLP:journals/ai/AlvianoDFPR23}. $US(P,O)$ denotes the set of unsatisfiable subsets of $P$ wrt\ $O$. An unsatisfiable subset $U \in US(P,O)$ is a \emph{minimal unsatisfiable subset} (MUS) of $P$ wrt\ $O$ iff for every $U' \subset U$, $U' \notin US(P,O)$.
An answer set $M$ of $P$ is minimal wrt $O \subset \mathcal{B_P}$ iff there exist no other answer set $M'$ of $P$ such that $(M \cap O) \subset (M' \cap O)$~\cite{DBLP:journals/ai/AlvianoDFPR23}.

More details on ASP are available in~\cite{DBLP:journals/ngc/GelfondL91,DBLP:journals/cacm/BrewkaET11,DBLP:series/synthesis/2012Gebser,DBLP:books/sp/Lifschitz19}.

\section{Linear Temporal Logic on Finite Traces}\label{sec:preliminaries}
Linear Temporal Logic on Finite Traces (\LTLf) \cite{DBLP:conf/ijcai/GiacomoV13} is a variant of Linear Temporal Logic (\LTL) \cite{DBLP:conf/focs/Pnueli77} where the models, called \emph{traces}, are assumed to be finite. An $\LTLf$ formula over the alphabet $\A$ is defined according to the following grammar:
$$\varphi ::= a \in \A \mid \lnot \varphi \mid \varphi \land \varphi  \mid  \LTLNext \varphi  \mid \varphi \LTLUntil \varphi $$

\noindent where $\LTLNext$ (``Next'') and $\LTLUntil$ (``Until'') are temporal operators. Standard propositional boolean ($\lor$, $\rightarrow$, $\leftrightarrow$, $\top$, $\bot$), and temporal (e.g., \emph{Eventually} \textsf{F}, \emph{Always} \textsf{G}, \textit{Weak Next} \textsf{X$_w$} and \textit{Release} \textsf{R}) shorthands are defined as usual.
An interpretation is a finite trace $\pi = \pi_0 \pi_1 \cdots \pi_{n-1}$, with $\pi_i \subseteq \A$. $|\pi| = n$ is the trace length, and $\pi(i) = \pi_i$ is the $i$-th state of $\pi$.
The satisfaction relation $\pi, i \models \varphi$ is defined inductively as follows:
\begin{itemize}
    \item[] $\pi, i \models a$ iff $a \in \pi(i)$
    \item[] $\pi, i \models \lnot \varphi$ iff $\pi, i \not\models \varphi$
    \item[] $\pi, i \models \varphi \land \psi$ iff $\pi, i \models \varphi$ and $\pi, i \models \psi$
    \item[] $\pi, i \models \LTLNext \varphi$ iff $i+1 < |\pi|$ and $\pi, i+1 \models \varphi$
    \item[] $\pi, i \models \varphi \LTLUntil \psi$ iff there exists $k$ with $i \le k < |\pi|$ such that $\pi, k \models \psi$ and for all $j$ with $i \le j < k$, $\pi, j \models \varphi$
\end{itemize}

$\pi$ is a \textit{model} of $\varphi$ if $\pi, 0 \models \varphi$, denoted as $\pi \models \varphi$. We denote by $\Mod\varphi$ the set of models of $\varphi$, that is $\Mod\varphi = \{\pi \in (2^\A)^*: \pi\models\varphi\}$.
The satisfiability problem for \LTLf, i.e., deciding whether there exists a trace $\pi$ such that $\pi \models \varphi$, is PSPACE-complete~\cite{DBLP:conf/ijcai/GiacomoV13}. Complexity drops to NP-complete restricting the problem to traces of bounded length~\cite{DBLP:journals/jair/FiondaG18}.

\section{Formal Explanations} Let $\mathcal{F}$ be a finite set of \emph{features}. Each feature $i \in \mathcal{F}$ is associated to a finite \emph{domain} $D_i$. The corresponding \emph{feature space} $\mathbb{F}$ is the Cartesian product of the feature domains, 
$\mathbb{F} = \prod_{i \in \mathcal{F}} D_i$. Each element $\alpha \in \mathbb{F}$ is a total assignment of values to the features, with $\alpha(i) \in D_i$ for every $i \in \mathcal{F}$. We refer to such an assignment as an \emph{instance}. A classifier is a total function $f: \mathbb{F} \mapsto \mathcal{K}$, where $\mathcal{K}$ is a finite set of class labels. Given $\mathcal{S} \subseteq \mathcal{F}$ and two instances $\alpha,\beta \in \mathbb{F}$, we say $\alpha$ and $\beta$ are \textit{compatible} in $\mathcal{S}$ and write $\agrees{\alpha}{\mathcal{S}}{\beta}$ if, and only if, $\alpha(i)=\beta(i) \text{ for every } i \in \mathcal{S}.$

An \emph{explanation problem} is a tuple $(f,\alpha,c)$, where $f: \mathbb{F} \mapsto \mathcal{K}$ is a classifier, $\alpha \in \mathbb{F}$ is an instance, and $c \in \mathcal{K}$ is a class label such that $f(\alpha)=c$. Two main types of formal explanations are in the literature: \textit{abductive} (AXp) and \textit{contrastive} (CXp) explanations. An abductive explanation for $(f,\alpha,c)$ answers the question \textit{Why was class $c$ assigned to $\alpha$?}, while a contrastive explanation addresses the question \textit{Why was $\alpha$ not assigned a class different from $c$?}.

\paragraph{Weak Abductive Explanations (WAXp).} We say that $\mathcal{X} \subseteq \mathcal{F}$ is a WAXp for $(f,\alpha,c)$ if the following holds:

$$\forall \beta \in \mathbb{F}.\ \agrees{\beta}{\mathcal{X}}{\alpha} \implies f(\beta) = c$$

\paragraph{Weak Contrastive Explanations (WCXp).} We say that $\mathcal{Y} \subseteq \mathcal{F}$ is a WCXp for $(f,\alpha,c)$ if the following holds:

$$\exists \beta \in \mathbb{F}.\ \agrees{\beta}{\mathcal{F}\setminus\mathcal{Y}}{\alpha} \land f(\beta) \ne c$$

Intuitively, an abductive explanation identifies a set of feature values that is sufficient to guarantee the prediction $c$, that is once the feature in $\mathcal{X}$ are fixed to the values they assume in $\alpha$, every compatible instance is still classified as $c$. Instead, a contrastive explanation identifies a set of features whose values can be changed to alter the prediction, that is, keeping all features outside $\mathcal{Y}$ fixed to the values in $\alpha$, there exists an assignment to the features in $\mathcal{Y}$ for which the classifier outputs a class different from $c$.
Subset-minimal WAXp and WCXp are called abductive explanations (AXp) and contrastive explanations (CXp), respectively. AXps and CXps are related by \emph{hitting-set duality}: every AXp is a minimal hitting set of the set of all CXps, and vice versa~\cite{joao}.

\section{AXp and CXp for \LTLf}
In this section, we reformulate abductive and contrastive explanations for \LTLf. In the rest of the paper, we fix a finite alphabet $\A$ and a trace length $n \in \mathbb{N}$. We assume all traces are sequences in $(2^\mathcal{A})^n$. 

\medskip
\noindent\textbf{Notation.} We adopt a set-based representation of traces based on \emph{temporal literals}. Let
$\Vars = \{a_i : a \in \A,\ 0 \le i < n\}$
denote the set of time-indexed atoms, where $a_i$ indicates that $a$ holds at time instant $i$. The set of temporal literals is defined as
$\Lits = \bigcup_{a_i \in \Vars} \{a_i, \lnot a_i\}.$
For every temporal literal $\ell \in \Lits$, we denote by $\base{\ell}$ its underlying time-indexed atom, i.e., $\base{a_i}=a_i$ and $\base{\lnot a_i}=a_i$. Similarly, $\opplit{\ell}$ denotes the opposite literal, that is, $\opplit{a_i}=\lnot a_i$ and $\opplit{\lnot a_i}=a_i$. A set of temporal literals $E \subseteq \Lits$ is said to be \emph{invalid} if it contains both a literal and its complement, \emph{valid} otherwise, and \emph{complete} if, for every $a_i \in \Vars$, exactly one of $a_i$ or  $\lnot a_i$ belongs to $E$.
A trace $\pi$ is represented by the valid, complete set of temporal literals:
$$\L{\pi} = \bigcup_i\left(\{a_i: a \in \pi(i)\} \cup \{\lnot a_i: a\not\in \pi(i)\}\right).$$

For instance, the violating trace $\pi_2$ of Figure~\ref{fig:motivating}, of length $n=2$, is represented by
$\L{\pi_2} = \{\mathit{req}_0, \mathit{req}_1, \mathit{low}_1\} \cup \{\lnot\mathit{deliver}_0, \lnot\mathit{low}_0, \lnot\mathit{charge}_0, \lnot\mathit{deliver}_1, \lnot\mathit{charge}_1\}$. Given a set of temporal literals $E$, two traces $\pi$ and $\tau$ are \emph{compatible} on $E$ if they agree on every literal in $E$, that is, $E \subseteq \L{\pi} \cap \L{\tau}$. Accordingly, we define the set of all traces compatible with $\pi$ on $E$ as:
$$\Comp{\pi}{E} = \{\tau \in (2^\A)^{|\pi|}: E \subseteq \L{\pi} \cap \L{\tau}\}.$$

Intuitively, $\Comp{\pi}{E}$ consists of all traces having the same length of $\pi$ that preserve the behavior of $\pi$ specified by $E$ in terms of events happening (positive literals) and not happening (negative literals). Since $E$ is not required to be complete, only the events explicitly described by its literals are fixed, while all remaining propositions at all time instants are left unconstrained.

\medskip
\noindent\textbf{Explanation problem.}
An \emph{explanation problem} is a tuple $(\pi,\varphi)$, where $\pi \in (2^\A)^n$ is a trace and $\varphi$ is an \LTLf formula. We now reformulate weak abductive and weak contrastive explanations using the notation introduced above. Since $\pi\not\models\varphi \Leftrightarrow \pi\models\lnot\varphi$, we focus on explanation settings $(\pi,\varphi)$ such that $\pi\not\models\varphi$ without loss of generality. The complementary case is handled by replacing $\varphi$ with $\neg\varphi$. Intuitively, a set $E \subseteq \L{\pi}$ identifies the temporal literals
that are \emph{locked}, i.e., forced to retain the truth value they have in $\pi$, while all literals outside $E$ are \emph{unlocked} and free to change: $\Comp{\pi}{E}$ is precisely the set of traces obtained by
locking $E$. An abductive explanation locks enough observations to make the violation unavoidable; a contrastive explanation unlocks enough observations to make satisfaction possible. 

\begin{definition}[WAXp]\label{def:waxp}
Let $(\pi,\varphi)$ be an explanation problem. We say that $E \subseteq \L{\pi}$ is a WAXp for $(\pi,\varphi)$ if:
    $\pi\not\models\varphi$ and $\Comp{\pi}{E} \cap \Mod{\varphi} = \emptyset$
\end{definition}

\begin{definition}[WCXp]\label{def:wcxp}
Let $(\pi,\varphi)$ be an explanation problem. We say that $E \subseteq \L{\pi}$ is a WCXp for $(\pi,\varphi)$ if:
$\pi\not\models\varphi$ and $\Comp{\pi}{\L{\pi} \setminus E} \cap \Mod{\varphi} \ne \emptyset$
\end{definition}

\begin{definition}[AXp, CXp]
Let $(\pi,\varphi)$ be an explanation problem. A set $E \subseteq \L{\pi}$
is an \emph{AXp} (resp.\ \emph{CXp}) for $(\pi,\varphi)$ if it is a WAXp
(resp.\ WCXp) and no $E' \subset E$ is a WAXp (resp.\ WCXp).
\end{definition}

\begin{example}
Consider $\pi_2$ and $\varphi$ from Figure~\ref{fig:motivating}. The set $E = \{\mathit{low}_1\}$ is an AXp for $(\pi_2,\varphi)$: locking $\mathit{low}_1$ is enough, since every trace of length $2$ in which $\mathit{low}$ holds at the last instant violates $\LTLAlways(\mathit{low} \rightarrow \LTLNext\mathit{charge})$, regardless of all other observations; and $E$ is minimal, as the empty set is not a WAXp. Dually, $\{\neg\mathit{deliver}_0, \mathit{req}_1, \mathit{low}_1\}$ is a CXp: unlocking these three literals admits a model of $\varphi$, where $\mathit{deliver}$ is added at instant $0$ and $\mathit{req}, \mathit{low}$ are removed at instant $1$, and no proper subset is a WCXp.
\end{example}

\section{Equivalence to Classifier Explanations}

In this section, we formalize the \LTLf-as-binary-classifier setting, establish a correspondence between classifier-based explanations and our reformulation, and prove that the two formulations are equivalent.

\subsection{Encoding \LTLf Formulas as Binary Classifiers} 
Let $\mathcal{F} = \Vars$ be binary features, each with domain $\{\bot, \top\}$--that is, we associate one Boolean feature with each time-indexed atom $a_i$--and let $\mathcal{K} = \{0, 1\}$. A trace $\pi \in (2^\A)^n$ is represented by the assignment $\asgn{\pi}$ defined by $\asgn{\pi}(a_i) = \top \Leftrightarrow a \in \pi(i)$. We denote by $\mathbb{F}$ the set of such assignments. Conversely, every assignment $\alpha \in \mathbb{F}$ introduces a trace $\traceof{\alpha}$, that is $\traceof{\alpha}(i) = \{a \in \A: \alpha(a_i) = \top\}$. 

Every \LTLf formula $\varphi$ induces a binary classifier $f_\varphi: \mathbb{F} \mapsto \{0,1\}$, defined by $$f_\varphi(\alpha)=
\begin{cases}
1 & \text{if }\traceof{\alpha}\models\varphi,\\
0 & \text{otherwise.}
\end{cases}$$ 

As in the previous section, we restrict our attention to instances satisfying $f_\varphi(\asgn{\pi})=0$, since
$\pi\not\models\varphi$ iff $\pi\models\neg\varphi$.

A set $\mathcal{X} \subseteq \mathcal{F}$ is a WAXp for $f_\varphi(\asgn{\pi}) = 0$ if:
$$\forall(\alpha \in \mathbb{F}).\ \agrees{\alpha}{\mathcal{X}}{\asgn{\pi}} \implies f_\varphi(\alpha) = 0.$$

A set $\mathcal{Y} \subseteq \mathcal{F}$ is a WCXp for $f_\varphi(\asgn{\pi}) = 0$ if:
$$\exists(\alpha \in \mathbb{F}).\ \agrees{\alpha}{\mathcal{F}\setminus\mathcal{Y}}{\asgn{\pi}}\ \land f_\varphi(\alpha) = 1.$$

Both formulations use the same names; they are distinguished by their arguments, as a WAXp (WCXp) for $f_\varphi(\asgn{\pi})=0$ is a set of \emph{features}, whereas a WAXp (WCXp) for $(\pi,\varphi)$, as per Definitions~\ref{def:waxp} and~\ref{def:wcxp}, is a set of \emph{temporal literals}. When needed, we call the former \emph{classifier-based} and the latter \emph{trace-based}.

To relate the two formulations, we introduce two mappings between sets of features and sets of temporal literals.
Given a trace $\pi$ and a set of feature $\mathcal{S} \subseteq \mathcal{F}$, we denote by $\litof{\pi}{\mathcal{S}}$ the set of literals of $\pi$ over the features in $\mathcal{S}$, that is,
$$\litof{\pi}{\mathcal{S}} = \{\ell \in \L{\pi}: \base{\ell} \in \mathcal{S}\}.$$

Conversely, for a set of temporal literals $E\subseteq\L{\pi}$, we denote by $\featof{E}$ the corresponding set of features:
$$\featof{E} = \{\base{\ell}: \ell \in E\}.$$

Observe that, by construction, for every trace $\pi$ and every $\mathcal S\subseteq\mathcal F$, the following properties hold:
\begin{enumerate}
    \item $\featof{\litof{\pi}{\mathcal{S}}} = \mathcal{S}$;
    \item $\litof{\pi}{\featof{E}} = E$, and
    \item $\L{\pi}\setminus\litof{\pi}{\mathcal S}=\litof{\pi}{\mathcal F\setminus\mathcal S}.$
\end{enumerate}

The equivalence proofs rely on the following observation: agreement between assignments corresponds to compatibility between the corresponding traces. More precisely, if two assignments agree on a set of features $\mathcal{S}$, then the corresponding traces share the temporal literals in $\litof{\pi}{\mathcal S}$.

\begin{proposition}\label{prop:compat-agreement}
Let $\pi,\tau$ be traces and let $E \subseteq \L{\pi}$. Then
$$\tau \in \Comp{\pi}{E} \quad\text{if and only if}\quad \agrees{\asgn{\tau}}{\featof{E}}{\asgn{\pi}}.$$
\end{proposition}

\begin{appendixonly}
\begin{proof}
$(\rightarrow)$ Assume that $\tau \in \Comp{\pi}{E}$. By definition, $E \subseteq \L{\pi}\cap\L{\tau}$. Hence, every temporal literal in $E$ belongs to both $\L{\pi}$ and $\L{\tau}$. Since $E \subseteq \L{\pi}$, every feature in $\featof{E}$ is assigned the same truth value by $\asgn{\pi}$ and $\asgn{\tau}$. Therefore, $\agrees{\asgn{\tau}}{\featof{E}}{\asgn{\pi}}$. $(\leftarrow)$ Assume that $\agrees{\asgn{\tau}}{\featof{E}}{\asgn{\pi}}$. Then, for every feature $a_i \in \featof{E}$, the assignments $\asgn{\pi}$ and $\asgn{\tau}$ are the same. Since $E \subseteq \L{\pi}$, each literal in $E$ also belongs to $\L{\tau}$. Thus, $E \subseteq \L{\pi}\cap\L{\tau}$, which implies that $\tau \in \Comp{\pi}{E}$.
\end{proof}
\end{appendixonly}

Using Proposition~\ref{prop:compat-agreement}, %
we can now prove the equivalence of the two notions of explanation. We start with WAXps. 

\begin{proposition}\label{prop:waxp-equivalence}
Let $(\pi,\varphi)$ be an explanation problem such that $\pi\not\models\varphi$, and let $\mathcal{X}\subseteq\mathcal{F}$.
Then $\mathcal{X}$ is a WAXp for $f_\varphi(\asgn{\pi})=0$ if and only if $\litof{\pi}{\mathcal{X}}$ is a WAXp for $(\pi,\varphi)$.
\end{proposition}
\begin{appendixonly}
\begin{proof}%
$(\rightarrow)$ Assume that $\mathcal{X}$ is a WAXp for
$f_\varphi(\asgn{\pi})=0$, and let $\tau\in\Comp{\pi}{\litof{\pi}{\mathcal{X}}}$. By Proposition~\ref{prop:compat-agreement},
$\agrees{\asgn{\tau}} {\featof{\litof{\pi}{\mathcal{X}}}}
{\asgn{\pi}}$. Since $\featof{\litof{\pi}{\mathcal{X}}}=\mathcal{X}$, we obtain $\agrees{\asgn{\tau}}{\mathcal{X}}{\asgn{\pi}}$. Because $\mathcal{X}$ is a WAXp, $f_\varphi(\asgn{\tau})=0$, which implies
$\tau\not\models\varphi$. Therefore, $\Comp{\pi}{\litof{\pi}{\mathcal{X}}}\cap\Mod{\varphi}=\emptyset$, and hence $\litof{\pi}{\mathcal{X}}$ is a WAXp for $(\pi,\varphi)$. $(\leftarrow)$ Assume that $\litof{\pi}{\mathcal{X}}$ is a WAXp for $(\pi,\varphi)$. Let $\alpha\in\mathbb{F}$ be such that $\agrees{\alpha}{\mathcal{X}}{\asgn{\pi}}$. Let $\tau=\traceof{\alpha}$. Since $\featof{\litof{\pi}{\mathcal{X}}}
= \mathcal{X}$, Proposition~\ref{prop:compat-agreement} implies that $\tau\in\Comp{\pi}{\litof{\pi}{\mathcal{X}}}$.
Because $\litof{\pi}{\mathcal{X}}$ is a WAXp for $(\pi,\varphi)$, $\tau\not\models\varphi$. Hence $f_\varphi(\alpha)=0$, which proves that $\mathcal{X}$ is a WAXp for $f_\varphi(\asgn{\pi})=0$.
\end{proof}
\end{appendixonly}

An analogous correspondence holds for WCXp.

\begin{proposition}\label{prop:wcxp-equivalence}
Let $(\pi,\varphi)$ be an explanation problem such that $\pi\not\models\varphi$ and let $\mathcal{Y} \subseteq \mathcal{F}$. Then $\mathcal{Y}$ is a WCXp for $f_\varphi(\asgn{\pi}) = 0$ if and only if $\litof{\pi}{\mathcal{Y}}$ is a WCXp for $(\pi,\varphi)$.
\end{proposition}

\begin{appendixonly}
\begin{proof} %
$(\rightarrow)$ Assume that $\mathcal{Y}$ is a WCXp for $f_\varphi(\asgn{\pi}) = 0$. Then there exists $\alpha \in \mathbb{F}$ such that $\agrees{\alpha}{\mathcal{F}\setminus\mathcal{Y}}{\asgn{\pi}}$ and $f_\varphi(\alpha)=1$. Let $\tau=\traceof{\alpha}$. 
By Proposition~\ref{prop:compat-agreement}, $\tau \in \Comp{\pi}{\litof{\pi}{\mathcal{F}\setminus\mathcal{Y}}}$ that is equivalent to $\tau\in\Comp{\pi}{\L{\pi}\setminus\litof{\pi}{\mathcal{Y}}}$ (property 3 above).
Moreover, $f_\varphi(\alpha)=1$ implies $\tau\models\varphi$. Hence $\Comp{\pi}{\L{\pi}\setminus\litof{\pi}{\mathcal{Y}}}\cap\Mod{\varphi}\ne\emptyset$, which proves that $\litof{\pi}{\mathcal{Y}}$ is a WCXp for $(\pi,\varphi)$. $(\leftarrow)$ Assume $\litof{\pi}{\mathcal{Y}}$ being a WCXp for $(\pi,\varphi)$. Then there exists $\tau \in \Comp{\pi}{\L{\pi}\setminus\litof{\pi}{\mathcal{Y}}}\cap\Mod{\varphi}$. By property 3 above, $\tau\in\Comp{\pi}{\litof{\pi}{\mathcal{F}\setminus\mathcal{Y}}}$. By Proposition~\ref{prop:compat-agreement}, $\agrees{\asgn{\tau}}{\featof{\litof{\pi}{\mathcal{F}\setminus\mathcal{Y}}}}{\asgn{\pi}}$. Since $\featof{\litof{\pi}{\mathcal{F}\setminus\mathcal{Y}}}=\mathcal{F}\setminus\mathcal{Y}$ (property 1), we have $\agrees{\asgn{\tau}}{\mathcal{F}\setminus\mathcal{Y}}{\asgn{\pi}}$. 
Finally, $\tau\models\varphi$ implies $f_\varphi(\asgn{\tau})=1$, and thus $\mathcal{Y}$ is a WCXp for $f_\varphi(\asgn{\pi})=0$.
\end{proof}
\end{appendixonly}

The previous propositions establish a bijective correspondence between the two formulations, i.e. each set $E \subseteq \L{\pi}$ that is a WAXp (WCXp, respectively) for $(\pi,\varphi)$ can be decoded into the set $\featof{E} \subseteq \mathcal{F}$ that is a WAXp (WCXp, respectively) for $f_\varphi(\asgn{\pi})=0$. The result holds since $\featof{\litof{\pi}{\mathcal{S}}}=\mathcal{S}$ and $\litof{\pi}{\featof{E}}=E$. Moreover, from the fact that both preserve subsets, we have that the statements also hold for AXp and CXp.

\begin{corollary}\label{cor:minimal-equivalence}
Let $(\pi,\varphi)$ be an explanation problem with $\pi\not\models\varphi$. Then $\mathcal{X}$ is an AXp (resp.\ CXp) for $f_\varphi(\asgn{\pi})=0$ if and only if $\litof{\pi}{\mathcal{X}}$ is an AXp (resp.\ CXp) for $(\pi,\varphi)$.
\end{corollary}

By AXp/CXp duality~\cite{joao} and since $\litof{\pi}{\cdot}$ is an intersection- and inclusion-preserving bijection between $\mathcal{F}$ and $\L{\pi}$, the following result holds.

\begin{corollary}[Duality]
Let $(\pi,\varphi)$ be an explanation problem with $\pi\not\models\varphi$. Then every AXp for $(\pi,\varphi)$ is a minimal hitting set of the set of all CXps for $(\pi,\varphi)$, and viceversa.
\end{corollary}

\section{Relationship with Temporal Causality}
Meng et al.~\cite{meng2025causal} recently proposed a causality-based framework for explaining finite counterexamples to \LTLf specifications. Their notion of explanation is based on Halpern's causality, that is a cause is a minimal set of interventions that transform a counterexample trace into a model of the specification. In this section we reformulate their framework using our notation and prove its formal relationship with our characterization of AXps and CXps.

Let $\pi$ be a trace, and let $\varphi$ be a formula. A \emph{candidate causal set} is a set $C \subseteq \L{\pi}$. Given $C \subseteq \L{\pi}$, the \emph{alteration trace} $\pi^C$ is a trace obtained by \emph{replacing} every literal of $\pi$ that occurs in $C$ with their opposite. More formally, if $C \subseteq \L{\pi}$, then $\L{\pi^C} = (\L{\pi} \setminus C) \cup \overline{C}$. A \emph{minimal cause for violation} is defined as follows:

\begin{definition} [Minimal cause for violation]
Let $\pi$ be a trace and $\varphi$ a formula with $\pi\not\models\varphi$. A set $C \subseteq \L{\pi}$ is a minimal cause for violation of $\varphi$ in counterexample $\pi$ if $\pi^C \models \varphi$ and no proper subset $C' \subset C$ is such that $\pi^{C'} \models \varphi$.
\end{definition}

Let $Alt_\pi(C) = \{\pi^X: X \subseteq C\}$ denote the set of alteration traces obtained by intervening on $\pi$ over subsets of $C$. The following lemma establishes the correspondence between alteration and compatible traces.

\begin{lemma} \label{lemma:tc}
Let $\pi$ be a trace and let $E \subseteq \L{\pi}$. Then $Alt_\pi(\L{\pi} \setminus E) = \Comp{\pi}{E}$.
\end{lemma}
\begin{appendixonly}
\begin{proof} %
$(\rightarrow)$ Let $\tau \in Alt_\pi(\L{\pi} \setminus E)$. Then there exists $X \subseteq \L{\pi} \setminus E$ such that $\tau = \pi^X$. By definition, $\L{\tau} = \L{\pi^X} = (\L{\pi} \setminus X) \cup \overline{X}$. Since $X\cap E=\emptyset$, every literal in $E$ is left unchanged. Hence, $E \subseteq \L{\pi}\cap \L{\tau}$ that implies $\tau \in \Comp{\pi}{E}$.
$(\leftarrow)$ Let $\tau \in \Comp{\pi}{E}$. Consider the set $X = \L{\pi} \setminus \L{\tau}$. Since $E\subseteq \L{\pi}\cap\L{\tau}$ we have that $X\cap E=\emptyset$. Therefore, $X \subseteq \L{\pi} \setminus E$. Moreover, $\tau = \pi^X$ because the intervention flips exactly the literals on which $\pi$ and $\tau$ differ. Hence, $\tau \in Alt_\pi(\L{\pi} \setminus E)$.
\end{proof}
\end{appendixonly}

Now we are ready to establish the relationships. 

\begin{proposition}\label{prop:causes}
    A set $C\subseteq \L{\pi}$ is a minimal cause for the violation of $\varphi$ if and only if $C$ is a CXp for $(\pi,\varphi)$.
\end{proposition}

\begin{appendixonly}
\begin{proof} %
    $(\rightarrow)$ Let $C$ be a minimal cause for violation of the formula $\varphi$ in the counterexample trace $\pi$. By definition, $\pi^C \models \varphi$.
    By Lemma~\ref{lemma:tc}, $Alt_\pi(C) = \Comp{\pi}{\L{\pi} \setminus C}$. Since $\pi^C \in Alt_\pi(C)$ and $\pi^C \models \varphi$, we have that $\Comp{\pi}{\L{\pi} \setminus C}\cap \Mod{\varphi}\neq\emptyset$. Hence, $C$ is a WCXp. Suppose that $C$ is not a CXp, then there exists a $C'\subset C$ that is also a WCXp. By Lemma~\ref{lemma:tc} $Alt_\pi(C') = \Comp{\pi}{\L{\pi}\setminus C'}$, hence there exists $X\subseteq C'$ such that $\pi^{X}\models \varphi$. Since $X\subseteq C'\subset C$, this contradicts the minimality of $C$.
    $(\leftarrow)$ Let $C$ be a CXp (hence also a WCXp), so $\Comp{\pi}{\L{\pi}\setminus C}\cap\Mod{\varphi}\neq\emptyset$. By Lemma~\ref{lemma:tc}, $\Comp{\pi}{\L{\pi} \setminus C} = Alt_\pi(C)$, hence there exists $X \subseteq C$ with $\pi^X \models \varphi$. Take such an $X$ that is subset-minimal (it exists, as $\L{\pi}$ is finite). Every $X' \subset X$ also satisfies $X' \subseteq C$, hence $\pi^{X'} \not\models \varphi$ by minimality of $X$; therefore $X$ is a minimal cause for violation. By direction $(\rightarrow)$, $X$ is a CXp; since $X \subseteq C$ and $C$ is a CXp, subset-minimality gives $X = C$. Therefore, $C$ is a minimal cause.
\end{proof}
\end{appendixonly}

Unlike \citeauthor{meng2025causal}, our approach makes AXp/CXp duality explicit and enables efficient cause computation.

\begin{table*}[ht!]

\centering
\caption{Per-instance AXp/CXp counts (top) and AXp/CXp explanation sizes in literals (bottom), by log. $N$: instances (C=complete, T=timeout). Column $Px$ refers to the $x$-th quantile of the distributions, with P50 being the median. Max., min. refers to the maximum and minimum observed statistics.}
\label{tab:results}
\resizebox{\textwidth}{!}{%
\begin{tabular}{lcccccccccccccccccccccc}
\toprule
& \multicolumn{2}{c}{$N$} & \multicolumn{10}{c}{AXp} & \multicolumn{10}{c}{CXp} \\
\cmidrule(lr){2-3} \cmidrule(lr){4-13} \cmidrule(lr){14-23}
& & & \multicolumn{5}{c}{Complete} & \multicolumn{5}{c}{Timeout} & \multicolumn{5}{c}{Complete} & \multicolumn{5}{c}{Timeout} \\
\cmidrule(lr){4-8} \cmidrule(lr){9-13} \cmidrule(lr){14-18} \cmidrule(lr){19-23}
Log & C & T & Min & P25 & P50 & P75 & Max & Min & P25 & P50 & P75 & Max & Min & P25 & P50 & P75 & Max & Min & P25 & P50 & P75 & Max \\
\midrule
\multicolumn{23}{l}{\textit{Per-instance counts}} \\
\midrule
bpi2020-domdecl & 81 & 8 & 1 & 10 & 21 & 82 & 441 & 567 & 578 & 588 & 647 & 834 & 1 & 1 & 3 & 5 & 19 & 18 & 23 & 28 & 35 & 70 \\
bpi2020-intdecl & 181 & 552 & 6 & 54 & 83 & 193 & 621 & 81 & 186 & 249 & 344 & 784 & 4 & 19 & 38 & 62 & 220 & 0 & 9 & 38 & 91 & 908 \\
bpi2020-permlog & 312 & 1166 & 4 & 40 & 177 & 443 & 1087 & 80 & 294 & 426 & 597 & 1253 & 2 & 15 & 51 & 143 & 366 & 0 & 9 & 70 & 223 & 885 \\
bpi2020-prepcost & 180 & 18 & 1 & 16 & 43 & 155 & 797 & 323 & 432 & 508 & 591 & 754 & 1 & 2 & 4 & 7 & 26 & 6 & 14 & 17 & 21 & 36 \\
bpi2020-reqpay & 79 & 1 & 1 & 6 & 12 & 62 & 937 & 829 & 829 & 829 & 829 & 829 & 1 & 1 & 2 & 4 & 26 & 88 & 88 & 88 & 88 & 88 \\
sepsis & 220 & 624 & 1 & 13 & 24 & 110 & 618 & 43 & 216 & 326 & 435 & 931 & 2 & 18 & 50 & 121 & 415 & 0 & 26 & 96 & 192 & 725 \\
\midrule
\multicolumn{23}{l}{\textit{Explanation sizes}} \\
\midrule
bpi2020-domdecl & 81 & 8 & 1 & 3 & 4 & 4 & 5 & 1 & 4 & 4 & 5 & 6 & 1 & 10 & 14 & 16 & 24 & 18 & 22 & 24 & 30 & 38 \\
bpi2020-intdecl & 181 & 552 & 1 & 4 & 6 & 7 & 10 & 1 & 7 & 8 & 10 & 29 & 2 & 10 & 14 & 16 & 22 & 6 & 14 & 18 & 20 & 40 \\
bpi2020-permlog & 312 & 1166 & 1 & 5 & 6 & 7 & 11 & 1 & 5 & 6 & 8 & 55 & 2 & 18 & 20 & 20 & 24 & 16 & 24 & 26 & 30 & 144 \\
bpi2020-prepcost & 180 & 18 & 1 & 3 & 4 & 4 & 6 & 1 & 4 & 4 & 4 & 6 & 1 & 14 & 18 & 20 & 30 & 20 & 24 & 24 & 26 & 36 \\
bpi2020-reqpay & 79 & 1 & 1 & 3 & 4 & 4 & 6 & 1 & 4 & 4 & 4 & 7 & 1 & 10 & 16 & 20 & 24 & 24 & 26 & 26 & 28 & 30 \\
sepsis & 220 & 624 & 1 & 5 & 6 & 7 & 11 & 1 & 7 & 9 & 10 & 54 & 1 & 10 & 12 & 14 & 22 & 4 & 14 & 16 & 18 & 50 \\
\bottomrule
\end{tabular}
}

\end{table*}

\section{Complexity}
We analyze the complexity of deciding whether a given set of temporal literals is a WAXp or WCXp. Trace length is fixed as part of the explanation setting and the decision problems amount to searching for models of \LTLf formulas over traces of bounded length. This reduces the complexity of satisfiability from PSPACE-complete to NP-complete.

\begin{proposition}\label{prop:wcxp-np}
Checking whether $E \subseteq \L{\pi}$ is a WCXp for an explanation problem $(\pi,\varphi)$ is NP-complete.
\end{proposition}
\begin{appendixonly}
\begin{proof}[Proof (Sketch)] %

[Membership.] A certificate is a trace $\tau\in (2^\A)^{|\pi|}$. One can check in polynomial time that both $\tau \in \Comp{\pi}{\L{\pi}\setminus E}$ and $\tau \models \varphi$. If both conditions hold, then $\Comp{\pi}{\L{\pi}\setminus E}\cap\Mod{\varphi}\neq\emptyset$ and  $E$ is a WCXp.

[Hardness.] Hardness is shown by a reduction from the Boolean satisfiability problem~\cite{garey1979computers}. Let $\phi$ be a CNF formula over the variables $X=\{x_1,...,x_n\}$,%
and let $q\not\in X$ be a fresh variable. Set $\A=X\cup\{q\}$ and $|\pi|=1$. We map $\phi$ to the instance $(\pi,\varphi,E)$ where:
$$\varphi = q\land \phi,
\qquad
\pi(0)=\emptyset,
\qquad
E=\L{\pi}.$$

Note that $\L{\pi}\setminus E=\emptyset$, so $\Comp{\pi}{\emptyset}$ contains every trace of length $1$ over $\A$; hence $E$ is a WCXp iff some length-$1$ trace satisfies $\varphi$. The conjunct $q$ ensures $\pi\not\models\varphi$, so $(\pi,\varphi)$ is a well-formed explanation problem.

($\rightarrow$) Let $M$ satisfy $\phi$, and set $\tau(0)=\{x_i\mid M(x_i)=\text{true}\}\cup\{q\}$. Then $\tau\models q\land\phi$, so $\Comp{\pi}{\L{\pi}\setminus  E}\cap\Mod{\varphi}\neq\emptyset$ and $E$ is a WCXp.
($\leftarrow$) Let $E$ be a WCXp, witnessed by $\tau\in\Comp{\pi}{\L{\pi}\setminus E}\cap\Mod{\varphi}$. Since $\varphi=q\land\phi$, we have $\tau\models\phi$; the assignment $M(x_i)=\text{true}$ iff $x_i\in\tau(0)$ satisfies $\phi$.
\end{proof}
\end{appendixonly}

The complexity of WAXp checking follows by duality. $E\subseteq\L{\pi}$ is a WAXp for $(\pi,\varphi)$ if and only if no trace in $\Comp{\pi}{E}$ satisfies $\varphi$, that is, $\L{\pi}\setminus E$ is \emph{not} a WCXp for $(\pi,\varphi)$. WAXp checking is thus the complement of WCXp checking, and Proposition~\ref{prop:wcxp-np} yields the following.

\begin{corollary}\label{prop:waxp-conp}
 Checking whether $E \subseteq \L{\pi}$ is a WAXp for an explanation problem $(\pi,\varphi)$ is coNP-complete.    
\end{corollary}

Given the coNP (resp.\ NP) complexity of checking whether a set is a WAXp (resp.\ WCXp), a deletion-based procedure that removes literals one by one computes an AXp (resp.\ CXp) with a polynomial (linear) number of oracle calls, which matches general results for AXps/CXps~\cite{DBLP:conf/aaai/0001I22}.

Next, we show that checking whether a set $E$ is an AXp (CXp) is DP-complete (coNP-complete) by reduction from MUS-checking (UNSAT) of a propositional formula.
Let $\phi$ be an unsatisfiable propositional formula in 3-CNF with clauses $c_1,\dots,c_m$, and let $R=\{r_1,\dots,r_m\}$ be fresh propositional symbols, one per clause. Set $c'_i = r_i \lor c_i$ and $\phi' = \bigwedge_{i=1}^{m} c'_i$. Over the alphabet $\A = \mathrm{vars}(\phi)\cup R$, the formula $\phi'$ is an \LTLf formula and $\pi$, with $\pi(0)=\emptyset$, is a trace of length one; since $\phi$ is unsatisfiable, $\pi\not\models\phi'$. Keeping $r_i$ false forces $c'_i$ to be satisfied through $c_i$; accordingly, we identify a set of clauses $S\subseteq\{c_1,\dots,c_m\}$ with the temporal literals $\lambda(S) = \{\lnot r_i : c_i\in S\}\subseteq\L{\pi}$. This yields the key correspondence: for every $S\subseteq\{c_1,\dots,c_m\}$,
$$
\Comp{\pi}{\lambda(S)}\cap\Mod{\phi'}\neq\emptyset
\iff
\textstyle\bigwedge_{c_i\in S} c_i \text{ is satisfiable.}(\ast)
$$
Recall that $S\subseteq\phi$ is an MUS if $S$ is unsatisfiable and every proper subset is satisfiable; $S$ is a \emph{maximal satisfiable subset} (MSS) if $S$ is satisfiable and $S\cup\{c\}$ is unsatisfiable for every $c\in\phi\setminus S$; and $S$ is an MCS if $\phi\setminus S$ is an MSS, i.e., a minimal set of clauses whose removal restores satisfiability.

\begin{proposition}\label{prop:mus-axp}
$S$ is an MUS of $\phi$ if and only if $\lambda(S)$ is an AXp for
$(\pi,\phi')$.
\end{proposition}
\begin{appendixonly}
\begin{proof} %
$(\rightarrow)$ Let $S$ be an MUS. By $(\ast)$, $S$ unsatisfiable gives $\Comp{\pi}{\lambda(S)}\cap\Mod{\phi'}=\emptyset$, so $\lambda(S)$ is a WAXp; and for every $S'\subset S$, $S'$ satisfiable gives, again by $(\ast)$, that $\lambda(S')$ is not a WAXp. These $\lambda(S')$ are exactly the proper subsets of $\lambda(S)$; hence $\lambda(S)$ is an AXp.
$(\leftarrow)$ Symmetric, reading each implication of $(\ast)$ in the
opposite direction.
\end{proof}
\end{appendixonly}

\begin{proposition}\label{prop:axp-dp}
Checking whether $E\subseteq\L{\pi}$ is an AXp for $(\pi,\varphi)$ is
DP-complete.
\end{proposition}
\begin{appendixonly}
\begin{proof}%
\emph{Membership.} $E$ is an AXp iff (a) $E$ is a WAXp and (b) no $E\setminus\{e\}$ is a WAXp. Condition (a) is in coNP (Corollary~\ref{prop:waxp-conp}); (b) is in NP, as a certificate is, for some $e\in E$, a trace in $\Comp{\pi}{E\setminus\{e\}}\cap\Mod{\varphi}$. Hence the problem is in DP.
\emph{Hardness.} By Proposition~\ref{prop:mus-axp}, since MUS-checking is DP-hard (MINIMAL UNSATISFIABILITY;~\cite{DBLP:journals/jcss/PapadimitriouW88,DBLP:journals/ai/Marques-SilvaJM17}).
\end{proof}
\end{appendixonly}

Checking a CXp, on the other hand, is coNP-complete:

\begin{proposition}
Checking whether $C\subseteq\L{\pi}$ is a CXp for $(\pi,\varphi)$ is
coNP-complete.
\end{proposition}
\begin{appendixonly}
\begin{proof}[Proof] %
[Membership.] By Proposition~\ref{prop:causes}, $C$ is a CXp if and only if (a) $\pi^C \models \varphi$ and (b) $\pi^{C'} \not\models \varphi$ for every $C' \subset C$. Condition (a) is a single model check on a fixed trace,
feasible in polynomial time~\cite{DBLP:journals/jair/FiondaG18}. Condition (b) is in coNP since a counterexample is a set $C' \subset C$ with $\pi^{C'} \models \varphi$, that can be checked in polynomial time. 

[Hardness.] We reduce from UNSAT. Let $\phi$ be a CNF formula over $X=\{x_1,\dots,x_k\}$. We assume w.l.o.g.\ that the all-true assignment does not satisfy $\phi$, this can be checked in polynomial time, and if it
does, we output a fixed instance that is not a CXp (e.g.\ any $(\pi,\varphi)$ with $\pi\models\varphi$). Set $\A = X\cup\{q\}$ with $q$ fresh, $|\pi|=1$, $\pi(0)=\emptyset$, and $\psi = q \land \big(\lnot(x_1\land\dots\land x_k)\rightarrow\phi\big)$. Since $q$ is false in $\pi$, $\pi\not\models\psi$, so $(\pi,\psi)$ is a well-formed explanation problem. Let $C=\L{\pi}$, we claim $C$ is a CXp for $(\pi,\psi)$ iff $\phi$ is unsatisfiable.
Since $\pi^C$ sets every atom to true, $\pi^C\models\psi$: $q$ holds and the antecedent $\lnot(\bigwedge_i x_i)$ is false. By Proposition~\ref{prop:causes}, $C$ is a CXp iff no proper subset $X'\subset C$ satisfies $\pi^{X'}\models\psi$; we show this holds iff $\phi$ is unsatisfiable.

$(\rightarrow)$ Suppose $\phi$ is unsatisfiable. Take any $X'\subset C$. If $q\notin X'$, then $q$ stays false and $\pi^{X'}\not\models\psi$. If $q\in X'$, then since $X'$ is proper some $x_i$ stays false, so the
antecedent $\lnot(\bigwedge_i x_i)$ holds and $\pi^{X'}\models\psi$ would require $\pi^{X'}\models\phi$; but $\phi$ is unsatisfiable, so no such $X'$ exists. Hence no proper subset satisfies $\psi$, and $C$ is a CXp.

$(\leftarrow)$ Suppose $\phi$ is satisfiable. By the w.l.o.g. assumption, it has a model $M$ different from all-true. Let $X'=\{q\}\cup\{x_i : M(x_i)= true\}$. Then $X'\subset C$, and $\pi^{X'}\models\psi$. Thus a proper subset satisfies $\psi$, and $C$ is not a CXp.
\end{proof}
\end{appendixonly}

\section{Computing AXps and CXps in ASP}
We build upon the approach of~\cite{DBLP:conf/aaai/FiondaIR26}, that describes an Answer Set Programming (ASP) logic program $\Pi^k(\varphi)$, such that its stable models correspond one-to-one to traces $\pi$ such that $|\pi| \le k$ and $\pi\models\varphi$. $\Pi^k(\varphi)$ is a typical guess-and-check ASP encoding, with a \emph{guess} program $\Pi^k_G(\varphi)$ that generates candidate traces, and a \emph{check} program $\Pi^k_C(\varphi)$ that discards traces that are not models of $\varphi$. 
More in detail, $\Pi^k_C(\varphi)$ is a stratified logic program that evaluates the syntactic tree of $\varphi$ in a bottom-up fashion over a guessed trace in $\Pi^k_G(\mathcal{A})$, discarding answer sets that encode traces that are not models of $\varphi$.

\begin{appendixonly}
We briefly describe the original encoding for bounded \LTLf satisfiability. 

\paragraph{Formulae \& traces} Consider a formula $\varphi$, and suppose each node in its syntax tree is assigned a unique integer. We can reify $\varphi$ into ASP facts by means of predicates $until/3$, $negate/2$, $conjunction/2$, $atomic/2$, $next/2$, $true/1$, where the first term always refers to the unique identifier of the subformula and the other terms refer to identifiers of the most immediate subformulae (children). Predicate names are self-explanatory, and correspond to the given temporal/propositional operator.

Wlog, the root of $\varphi$ is identified by 0. We denote this set of facts as $[\varphi]$. For example, the formula $\lnot a \Until (\Next b)$ would be encoded as:
\texttt{until(0,1,2). negate(1,3). atomic(3,a). next(2,4). atomic(4,b). root(0).}

A trace $\pi$ can be encoded by means of facts $trace(t,a)$ if $a \in \sigma_t$, and $time(t)$ if $t < |\pi|$, where $trace(t,a)$ holds if $a \in \pi(t)$, and $time(i)$ if $|\pi| = k$ and $0 \le i \le k-1$. 
For example, the trace $\{a,b\}\cdot\{\}\cdot\{c\}$ is encoded as:
\texttt{time(0..2). trace(0,a). trace(0,b). trace(2,c).}
We denote this set of facts as $[\pi]$.

\paragraph{Semantics of temporal operators} The ASP encoding of temporal operators makes use of predicates with signature $holds/2$, $trace/2$, and $time/1$. An atom $holds(t,x)$ models that the subformula rooted in $x$ is true at time $t$. The atom $time(i)$ states that the trace has an $i$-th state, i.e. it has length at least $i$. The following logic program $\Pi_{LTLf}$ models the semantics of temporal operators:

\footnotesize{
\begin{verbatim} 
holds(T,X) :- atomic(X,A), trace(T,A).
holds(T,X) :- next(X,F), holds(T+1,F).
holds(T,X) :- conjunction(X,_), time(T), 
    holds(T,F): conjunction(X,F).
holds(T,X) :- negate(X,F), time(T), not holds(T,F).
holds(T,X) :- until(X,F,G), holds(T,G).
holds(T,X) :- until(X,F,G), holds(T,F), holds(T+1,X).
sat :- holds(0,X), root(X).
\end{verbatim}
}

Each rule closely mirrors the definition of the \emph{next normal form} transformation of each temporal operator, so correctness follows from previous results~\cite{DBLP:journals/jair/FiondaG18}; since this is a locally stratified logic program~\cite{DBLP:series/synthesis/2012Gebser}, it admits a unique answer set containing $sat$ iff $\pi \models \varphi$.

\paragraph{Search for a model}\label{subsec:search}
The above encoding evaluates the formula $\varphi$ over a trace $\pi$, so it is sufficient to replace the set of facts $[\pi]$ by some choice rules that generates all possible traces to lift it for checking bounded satisfiability. Consider the following logic program $\Pi_{Search}^k$:

\footnotesize{
\begin{verbatim}
symbol(X) :- atomic(_,X).
{ last_instant(T): T=0..k-1 } = 1.
time(0..T) :- last_instant(T).
{ trace(T,A): symbol(A) } :- time(T).
:- not sat.
\end{verbatim}
}

Above, \texttt{k} is a runtime \textsc{clingo} constant which denotes the search bounds. The first choice rule yields the length of the guessed trace, $last\_instant/1$, whereas the second one generates all possible states (given symbols that appear in $\varphi$), for each time point, thus all possible traces of length up to $k$. 
The constraint discards traces which are not models of $\varphi$. 
If the above logic program is satisfiable, projecting answer sets on predicates $trace/2$, $time/1$ allows one to obtain satisfying traces.

Thus, for the purposes of our paper, we have that $\Pi^k_G(\varphi) = \Pi^k_{Search}$, and $\Pi^k_C(\varphi) = [\varphi] \cup \Pi_{LTLf}$.
\end{appendixonly}

To apply our approach, we require to edit on $\Pi^k_G(\varphi)$, so that we fix the input trace $\pi$ as the unique model of the guess program.
This is done as follows.

Let $\pi \in (2^\mathcal{A})^t$. The program $G(\pi)$ containing a pair of rules of the form:

{\footnotesize \begin{verbatim} 
    trace(t,a) :- not ntrace(t,a).
    ntrace(t,a) :- not trace(t,a).
\end{verbatim}}

\noindent for each $(a, t) \in \mathcal{A} \times \{0, \dots, |\pi|-1\}$, guesses the possible interventions on an input trace $\pi$, that is an arbitrary trace of length $k$ by Lemma~\ref{lemma:tc} with $E=\emptyset$. 

Let's introduce the set of constraints $C(S) = \{\texttt{:- not trace(t,}a\texttt{)}:  a_t \in S\} \cup \{\texttt{:- not ntrace(t,}a\texttt{)}: \lnot a_t \in S\}$. 
Intuitively, these constraints \emph{lock} certain events in the models of $G(\pi)$, playing the same role as $\Comp{\pi}{\cdot}$.

\begin{proposition}
Let $\pi$ be a trace over $\mathcal{A}$, $S \subseteq \L{\pi}$. Models of $G(\pi) \cup C(S)$ are in one-to-one correspondence with $\Comp{\pi}{S}$. 
\end{proposition}
\begin{appendixonly}
\begin{proof}[Proof (Sketch).] By construction observing that $G(\pi)$ generates $(2^\mathcal{A})^t$ and constraints in $C(S)$ force the occurrence/absence of the temporal literals $\ell_t \in S$.
\end{proof}
\end{appendixonly}

The logic program $\Pi^k_C(\varphi)$ of the original encoding of~\cite{lpnmmr22} models the satisfaction of $\varphi$, thus:

\begin{proposition}
Answer sets of $G(\pi) \cup C(S) \cup \Pi^k_C(\varphi)$ are in one-to-one correspondence with traces in $\Comp{\pi}{S} \cap \Mod{\varphi}$. 
\end{proposition}

Note also that whenever $S = \L{\pi}$, the program $G(\pi) \cup C(S)$ admits an unique model, that is the trace $\pi$. Thus, if $\pi\not\models\varphi$, we have that $G(\pi) \cup C(\L{\pi}) \cup \Pi^k_C(\varphi)$ is incoherent. Consequently, it admits unsatisfiable cores and correction sets~\cite{DBLP:journals/ai/AlvianoDFPR23}. 
In particular, we are interested in MUSes and MCSes in terms of the constraints in $C(S)$. 

Let $O = \{o(t,a): a \in \mathcal{A}, 0 \le t < |\pi|\}$ be a set of fresh objective atoms. Let $C^*(S) = \{\texttt{:- o(t,a), not trace(t,}a\texttt{)}: a_t \in S\} \cup \{\texttt{:- o(t,a), not ntrace(t,}a\texttt{)}: \lnot a_t \in S\}$. Intuitively, each $o(t,a) \in O$ acts as a \emph{selector} to include/exclude a constraint in $C(S)$, effectively \emph{locking} certain temporal literals to assume the same value as in $\pi$. 

\begin{proposition}
Let $U \subseteq O$. Let $\Pi = G(\pi) \cup \Pi^k_C(\varphi) \cup C^*(\L{\pi})$. We have that $enforce(\Pi,O,U)$ is incoherent if and only if $S = \{a_t: o(t,a) \in U, a \in \pi(t)\} \cup \{\lnot a_t: o(t,a) \in U, a\not\in \pi(t)\}$ is a WAXp.
\end{proposition}
\begin{appendixonly}
    
\begin{proof}[Proof (Sketch).]
Follows from observing that including or excluding objective atoms in $U$ respectively forces or forbids the corresponding temporal literals to match $\pi$.
\end{proof}
\end{appendixonly}

Notice this also implies the results holds for MUSes and AXps, since those are minimal unsatisfiable cores and minimal WAXps. The analogous for MCSes follows trivially from minimal hitting set dualities between MUSes/MCSes in ASP programs~\cite{DBLP:journals/ai/AlvianoDFPR23}, and AXps/CXps~\cite{joao}:

\begin{proposition}
Let $U \subseteq O$. Let $\Pi = G(\pi) \cup \Pi^k_C(\varphi) \cup C^*(\L{\pi})$. We have that $enforce(\Pi,O,O\setminus U)$ is coherent if and only if $S = \{a_t: o(t,a) \in U, a \in \pi(t)\} \cup \{\lnot a_t: o(t,a) \in U, a\not\in \pi(t)\}$ is a WCXp.
\end{proposition}

\section{Experiments}
We evaluate the effectiveness of MUSes/MCSes enumeration for computing AXps and CXps of \LTLf specifications. 

\paragraph{Implementation.} Our implementation is a Python script that parses explanation problems of the form $(\pi,\varphi)$ into a logic program, and then applies the MARCO~\cite{DBLP:journals/constraints/LiffitonPMM16} algorithm for ASP~\cite{DBLP:journals/ai/AlvianoDFPR23}, to jointly enumerate MUSes (AXps) and MCSes (CXps). We use the \texttt{clingo}~\cite{DBLP:journals/tplp/GebserKKS19} ASP solver.

\begin{appendixonly}
\paragraph{Execution Environment.} Experiments were run on a machine equipped with a 13th Gen Intel(R) Core(TM) i7-13700 CPU and 32 GB of RAM, with 16 physical cores and 24 logical CPUs, running Ubuntu 22.04. In the experiments, we set a CPU-time limit of 10 seconds, a wall-clock timeout of 20 seconds, and a memory limit of 2 GB of RAM.    
\end{appendixonly}

\paragraph{Data.} Our benchmark is built from \emph{Declare}~\cite{DBLP:journals/ife/AalstPS09}, a
declarative process modeling language based on a set of templates whose semantics is given in \LTLf. In this setting, process discovery algorithms learn Declare models---that one can translate in \LTLf formulas---from event logs, i.e., multi-sets of traces recording real process executions. We choose Declare because it is a major application of \LTLf and, to the best of our knowledge, the only domain providing both real-world traces and \LTLf formulas that are learned from data rather than hand-crafted.  We consider a collection of 6 real-world event logs (from the BPI Challenge 2020~\cite{bpichallenge2020}, and the Sepsis event log~\cite{DBLP:conf/emisa/MannhardtB17}). For each event log, we learn a Declare model using the Declare4Py library~\cite{donadello2022declare4py}, which we translate into an \LTLf formula $\varphi$. For each \emph{variant} $\pi$ (i.e., unique trace) in the event log, we generate one explanation problem: $(\pi,\varphi)$
when $\pi\not\models\varphi$, and $(\pi,\lnot\varphi)$
when $\pi\models\varphi$. We obtain a total of $3467$ explanation problems, and we discard 45 where the corresponding \LTLf formula is incoherent--which would admit a single (empty) AXp and no CXps. This yields a total of 3422 instances.

\paragraph{Results.}
Table~\ref{tab:results} reports statistics on the number, size of AXps and CXps found by our approach, partitioned by corresponding event log. Top-half reports the number of AXps/CXps found; bottom-half their size. We distinguish between \emph{complete} instances, where enumeration terminates and \emph{incomplete} instances where it does not terminate\footnote{MARCO is an anytime algorithm, outputs AXps/CXps as soon as they are found, and terminates when \emph{all} CXps/AXps are found.}. The approach is effective, as in a short timeout we are able to complete about 30\% of the instances, while on the remaining \emph{incomplete} instances we are still able to output several CXps and AXps, with a few exceptions (155) where we find zero CXps. However, this is due to the fact that the enumeration algorithm jointly enumerates CXps and AXps, and those instances contain a large amount of AXps. Additionally, we report no memory out errors, and an average (median) time-to-first AXp of $0.091 \pm 0.032$s (0.09s) and an average time-to-first CXp of $1.05 \pm 1.58$s (0.41s).

\paragraph{Comparison with~\citet{meng2025causal}.} 
    The only existing tool that can compute our CXps is \texttt{caupybara}~\cite{meng2025causal}, which computes causes (equivalent to CXps) by brute-force. Candidate sets of increasing size (up to a user-provided upper bound on size) are enumerated, and checked for being a cause. As one may notice in Table~\ref{tab:results}, CXps of large size are not unlikely in real-world event logs. Thus, it is expected the approach scales poorly, as the size of causes increases. We highlight this limitation with a synthetic formula. Let $p(i) = \LTLAlways(a_i \rightarrow \LTLWNext\, b_i)$, $P(N) = \bigwedge_{i=1}^{N} p(i)$, and $\tau_N = \{a_1\}\{a_2\}\cdots\{a_N\}\{\}$, where the final empty state ensures the last weak next is not trivially satisfied. Every conjunct $p(i)$ is violated, so $\tau_N \not\models P(N)$, and the instance has exactly $N$ AXps of size~2 (each pairing the presence of $a_i$ with the absence of $b_i$) and, by hitting-set duality, $2^N$ CXps of size $N$. Altering the trace to satisfy $P(N)$ must satisfy each $p(i)$, choosing for each $i$ whether to remove $a_i$ or add $b_i$. Thus, one needs to enumerate up to $N$. Note that the conjunct $p(i)$ is not a ``pathological'' formula, but a rather common specification pattern in \LTLf specifications (``every time $a_i$ happens, either the trace concludes or $b_i$ happens next''). In particular, this is known as \emph{Chain Response} in Declare (with a strong next replacing the weak).
    
    Figure~\ref{fig:scalability} reports the time to compute one and all CXps as $N$ grows, under the same 10-second budget as before, with \texttt{caupybara} run at $\mathit{bound}=N$. \texttt{caupybara} exceeds the budget beyond $N{=}4$, \emph{even for a single CXp}. Our approach, instead, manages to enumerate all CXps up to $N{=}12$, while also computing AXps.  
    A CXp-only variant of our approach (enumerating \emph{minimal stable models}~\cite{DBLP:journals/ai/AlvianoDFPR23} to obtain MCSes/CXps, rather than applying MARCO-style joint MUS/MCS enumeration) matches \texttt{caupybara}'s task exactly, and finds a first CXp in under $0.1$s, but performs worse in enumeration, completing only up to $N{=}10$, overall slowing down in the CXp enumeration task. 
    
\begin{figure}[t]
    \centering
    \includegraphics[width=0.70\linewidth]{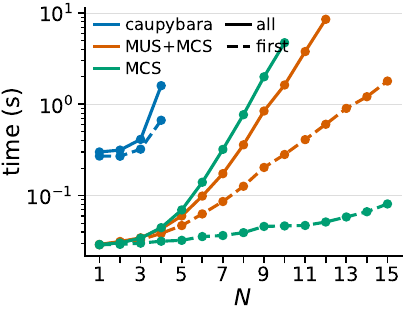}
    \caption{Time to the first and to all CXps on $P(N)$, $\tau_N$, for \texttt{caupybara}, our full approach, and a CXp-only variant. Dashed: first CXp; solid: all CXps. 10 seconds timeout.}
    \label{fig:scalability}
\end{figure}

\section{Conclusion}
In this paper, we specialize abductive and contrastive explanations to \LTLf. We provide an idiomatic definition that side-steps the classifier view the original definitions require, prove the two equivalent, and prove contrastive explanations equivalent to a recently-proposed notion of temporal causality. This characterization lets us enumerate explanations efficiently by exploiting MUS-MCS duality, which we instantiate in Answer Set Programming. Experiments show the approach is effective on real-world event logs.

\bibliography{aaai2026}

\appendix

\end{document}